\documentclass[journal]{IEEEtran}

\usepackage{cite}
\usepackage[cmex10]{amsmath}
\usepackage{amssymb}
\usepackage{amsthm}
\usepackage{mathrsfs}
\usepackage{bm}
\usepackage[mathscr]{eucal}
\usepackage{amssymb,amsmath,amsthm,amsfonts,latexsym}
\usepackage{amsmath,graphicx,bm,xcolor,url}
\usepackage{graphicx}
\usepackage{latexsym}
\usepackage{CJK}
\usepackage{indentfirst}
\usepackage{geometry}
\usepackage{psfrag}
\usepackage{setspace}
\usepackage{algorithmic}
\usepackage{algorithmic, cite}
\usepackage{algorithm}
\usepackage{array}
\usepackage{mdwmath}
\usepackage{mdwtab}
\usepackage{eqparbox}
\usepackage{url}
\usepackage{epstopdf}
\usepackage{epsfig,epsf,psfrag}
\usepackage{fixltx2e}
\usepackage{verbatim}
\usepackage{textcomp}
\usepackage{psfrag} 

\usepackage{subfigure} 
\usepackage{caption}

\theoremstyle{plain}
\newtheorem{mythe}{Theorem}
\theoremstyle{remark}

\theoremstyle{plain}

\theoremstyle{remark}

\theoremstyle{plain}

\theoremstyle{remark}

\theoremstyle{remark}

\theoremstyle{remark}

\theoremstyle{remark}

\theoremstyle{remark}

\theoremstyle{remark}

\catcode`~=11 \def\UrlSpecials{\do\~{\kern -.15em\lower .7ex\hbox{~}\kern .04em}} \catcode`~=13

\allowdisplaybreaks[4]

\newcommand{\calC}{\mathcal{C}}

\newcommand{\calG}{\mathcal{G}}

\newcommand{\calL}{\mathcal{L}}

\newcommand{\calN}{\mathcal{N}}
\newcommand{\calO}{\mathcal{O}}

\newcommand{\calR}{\mathcal{R}}

\newcommand{\bg}{\mathbf{g}}
\newcommand{\bG}{\mathbf{G}}

\newcommand{\bp}{\mathbf{p}}
\newcommand{\bP}{\mathbf{P}}

\newcommand{\bbE}{\mathbb{E}}

\DeclareMathAlphabet{\mathbsf}{OT1}{cmss}{bx}{n}
\DeclareMathAlphabet{\mathssf}{OT1}{cmss}{m}{sl}

\DeclareSymbolFont{bsfletters}{OT1}{cmss}{bx}{n}
\DeclareSymbolFont{ssfletters}{OT1}{cmss}{m}{n}
\DeclareMathSymbol{\bsfGamma}{0}{bsfletters}{'000}
\DeclareMathSymbol{\ssfGamma}{0}{ssfletters}{'000}
\DeclareMathSymbol{\bsfDelta}{0}{bsfletters}{'001}
\DeclareMathSymbol{\ssfDelta}{0}{ssfletters}{'001}
\DeclareMathSymbol{\bsfTheta}{0}{bsfletters}{'002}
\DeclareMathSymbol{\ssfTheta}{0}{ssfletters}{'002}
\DeclareMathSymbol{\bsfLambda}{0}{bsfletters}{'003}
\DeclareMathSymbol{\ssfLambda}{0}{ssfletters}{'003}
\DeclareMathSymbol{\bsfXi}{0}{bsfletters}{'004}
\DeclareMathSymbol{\ssfXi}{0}{ssfletters}{'004}
\DeclareMathSymbol{\bsfPi}{0}{bsfletters}{'005}
\DeclareMathSymbol{\ssfPi}{0}{ssfletters}{'005}
\DeclareMathSymbol{\bsfSigma}{0}{bsfletters}{'006}
\DeclareMathSymbol{\ssfSigma}{0}{ssfletters}{'006}
\DeclareMathSymbol{\bsfUpsilon}{0}{bsfletters}{'007}
\DeclareMathSymbol{\ssfUpsilon}{0}{ssfletters}{'007}
\DeclareMathSymbol{\bsfPhi}{0}{bsfletters}{'010}
\DeclareMathSymbol{\ssfPhi}{0}{ssfletters}{'010}
\DeclareMathSymbol{\bsfPsi}{0}{bsfletters}{'011}
\DeclareMathSymbol{\ssfPsi}{0}{ssfletters}{'011}
\DeclareMathSymbol{\bsfOmega}{0}{bsfletters}{'012}
\DeclareMathSymbol{\ssfOmega}{0}{ssfletters}{'012}

\newcommand{\tilQ}{\widetilde{Q}}

\newcommand{\tilr}{\widetilde{r}}
\newcommand{\tilR}{\widetilde{R}}

\newcommand{\tilW}{\widetilde{W}}

\newcommand{\hatX}{\widehat{X}}

\newcommand{\barP}{\bar{P}}

\newcommand{\balpha}{\bm{\alpha}}

\newcommand{\btau}{\bm{\tau}}

\newcommand{\bpsi}{\bm{\psi}}

\def\norm#1{\left\| #1 \right\|}
\def\norm2#1{\left\| #1 \right\|_2}
\def\norm22#1{\left\| #1 \right\|_2^2}

\newcommand{\eqa}{\stackrel{(a)}{=}}

\newcommand{\leb}{\stackrel{(b)}{\le}}
\newcommand{\lec}{\stackrel{(c)}{\le}}

\newcommand{\qednew}{\nobreak \ifvmode \relax \else
      \ifdim\lastskip<1.5em \hskip-\lastskip
      \hskip1.5em plus0em minus0.5em \fi \nobreak
      \vrule height0.75em width0.5em depth0.25em\fi}

\usepackage{float}
\usepackage{cite}
\usepackage{times,amsmath,color,amssymb,graphicx,epsfig,multirow,float,algorithm,algorithmic}

\title{Optimal Resource Allocation in Full-Duplex Ambient Backscatter Communication Networks for Wireless-Powered IoT}

\author{Gang~Yang, \emph{Member, IEEE}, Dongdong Yuan, Ying-Chang~Liang, \emph{Fellow, IEEE}, \\ Rui Zhang, \emph{Fellow, IEEE}, and Victor C. M. Leung, \emph{Fellow, IEEE}
\thanks{G.~Yang and D. Yuan are with the National Key Laboratory of Science and Technology on Communications, and Center for Intelligent Networking and Communications (CINC), University of Electronic Science and Technology of China (UESTC), Chengdu 611731, China (e-mails: yanggang@uestc.edu.cn, 201721260420@std.uestc.edu.cn).}
\thanks{Y.-C. Liang is with Center for Intelligent Networking and Communications (CINC), University of Electronic Science and Technology of China (UESTC), Chengdu 611731, China (e-mail: liangyc@ieee.org). (\emph{Corresponding author: Y.-C. Liang.})}
\thanks{R. Zhang is with the Department of Electrical and Computer Engineering, National University of Singapore, 117583, Singapore (e-mail: elezhang@nus.edu.sg).}
\thanks{V. C. M. Leung is with the Department of Electrical and Computer Engineering,
the University of British Columbia, Vancouver, BC, V6T 1Z4, Canada (e-mail: vleung@ece.ubc.ca).}}

\begin{document}
 \maketitle

\begin{abstract}
This paper considers an ambient backscatter communication (AmBC) network in which a full-duplex access point (FAP) simultaneously transmits downlink orthogonal frequency division multiplexing (OFDM) signals to its legacy user (LU) and receives uplink signals backscattered from multiple BDs in a time-division-multiple-access manner. To maximize the system throughput and ensure fairness, we aim to maximize the minimum throughput among all BDs by jointly optimizing the backscatter time and reflection coefficients of the BDs, and the FAP's subcarrier power allocation, subject to the LU's throughput constraint, the BDs' harvested-energy constraints, and other practical constraints. For the case with a single BD, we obtain closed-form solutions and propose an efficient algorithm by using the Lagrange duality method. For the general case with multiple BDs, we propose an iterative algorithm by leveraging the block coordinated decent and successive convex optimization techniques. We further show the convergence performances of the proposed algorithms and analyze their complexities. In addition, we study the throughput region which characterizes the Pareto-optimal throughput trade-offs among all BDs. Finally, extensive simulation results show that the proposed joint design achieves significant throughput gain as compared to the benchmark schemes. 
\end{abstract}


\section{Introduction}
Internet of Things (IoT) is a key application paradigm for the forthcoming fifth-generation (5G) and future wireless communication systems. IoT devices in practice have strict limitations on energy, cost, and complexity, thus it is highly desirable to design energy- and spectrum-efficient communication technologies \cite{StankovicJIoT14,YangZHanIoTJ18}. Recently, ambient backscatter communication (AmBC) has emerged as a promising candidate to fulfill such demand. On one hand, AmBC enables wireless-powered backscatter devices (BDs) to modulate their information symbols over ambient radio-frequency (RF) carriers (e.g., WiFi, TV, or cellular signals) without using any costly and power-hungry RF transmitter~\cite{ABCSigcom13}. On the other hand, no dedicated spectrum is needed for AmBC due to the spectrum sharing between the backscatter transmission and the ambient transmission~\cite{KangLiangICC17}.

The existing AmBC systems can be divided into three categories, namely the traditional AmBC (TABC) system with separated backscatter receiver and ambient transmitter (and its legacy\footnote{Hereinafter, the term ``legacy'' refers to any existing wireless communication systems such as WiFi.} receiver) \cite{WiFiBackscatter14,TurbochargingABCSigcom14,QianGaoAmBCTWC16,FSBackscatterSigcomm16, WangSmithFMBackscatter17,YangLiangZhangPeiTCOM17,DarsenaVerdeTCOM17,KangLiangICC17,HoangNiyatoAmBCTCOM17, ShahKaeWonChoiIoTJ18,ShenAthalyeDjuricIoTJ16}, the cooperative AmBC (CABC) system with co-located backscatter receiver and  legacy receiver \cite{YangLiangZhangIoTJ18,DuanRuttikTVT17,LongYangGC17}, and the full-duplex AmBC (FABC) system with co-located backscatter receiver and ambient transmitter \cite{BackFiSigcom15,DarsenaVerdeTCOM17}.

The TABC systems are most studied in the literature \cite{WiFiBackscatter14,TurbochargingABCSigcom14,QianGaoAmBCTWC16,FSBackscatterSigcomm16, WangSmithFMBackscatter17,YangLiangZhangPeiTCOM17,DarsenaVerdeTCOM17,KangLiangICC17,HoangNiyatoAmBCTCOM17, ShahKaeWonChoiIoTJ18,ShenAthalyeDjuricIoTJ16}. One of the key challenges for TABC systems is the strong direct-link interference from the ambient transmitter received at the backscatter receiver. Frequency-shifting method is proposed in \cite{FSBackscatterSigcomm16, WangSmithFMBackscatter17} to avoid the direct-link interference, while in~\cite{YangLiangZhangPeiTCOM17}, the direct-link interference is cancelled out through using the specific feature of the ambient signals.
There are also studies on TABC system performance and resource allocations \cite{DarsenaVerdeTCOM17, KangLiangICC17, HoangNiyatoAmBCTCOM17, ShahKaeWonChoiIoTJ18}. For example, in \cite{KangLiangICC17}, a TABC system is modelled from a spectrum sharing perspective, and the ergodic capacity of the secondary backscatter system is maximized. In~\cite{DarsenaVerdeTCOM17}, the capacity bounds for backscatter communication are derived for a TABC system, under the assumption that the backscatter receiver knows legacy symbols.


In CABC systems, the signals from the ambient transmitter are recovered at the backscatter receiver instead of being treated as interference \cite{YangLiangZhangIoTJ18,DuanRuttikTVT17,LongYangGC17}. In particular, the optimal maximum-likelihood detector, suboptimal linear detectors, and the successive interference-cancellation based detectors are derived in \cite{YangLiangZhangIoTJ18}. In \cite{DuanRuttikTVT17}, the sum rate of the backscatter communication and the legacy communication is analyzed under both perfect and imperfect channel state information for a CABC system with multiple antennas at each node. In \cite{LongYangGC17}, the transmit beamforming is optimized to maximize the sum rate of a CABC system in which the ambient transmitter is equipped with multiple antennas.

In FABC systems, the backscatter receiver and ambient transmitter are collocated, thus the signals from the ambient transmitter can be cancelled out~ \cite{BackFiSigcom15,DarsenaVerdeTCOM17}. The authors in~\cite{DarsenaVerdeTCOM17} analyze the capacity performances of both the backscatter communication and the legacy communication for an FABC system over OFDM carriers, and obtain the asymptotic capacity bounds in closed form when the number of subcarriers is sufficiently large. The authors in \cite{BackFiSigcom15} build an FABC system prototype in which the WiFi access point (AP) decodes the received backscattered signal while simultaneously transmitting WiFi packages to its legacy client. However, only a single BD is considered in \cite{DarsenaVerdeTCOM17} and \cite{BackFiSigcom15}, which simplifies the analysis and implementation but limits the applicability in practice. 

 The aforementioned prior works mainly focus on the transceiver design and hardware prototyping for various single-BD AmBC systems. To our best knowledge, the existing literature still lacks fundamental analysis and performance optimization for a general FABC system with multiple BDs.

 In this paper, we consider a full-duplex AmBC network (F-ABCN) over ambient OFDM carriers as shown in Fig.~\ref{fig:Fig1}, consisting of a full-duplex access point (FAP) with two antennas for simultaneous signal transmission and reception, respectively, a legacy user (LU), and multiple BDs. The FAP transmits dowlink signal which not only carries information to the LU but also transfers energy to the BDs; while at the same time all BDs perform uplink information transmission via backscattering in a time-division-multiple-access (TDMA) manner. The backscattered signal in general interferes with the LU's received information signal directly from the FAP. Thus, this proposed F-ABCN differs from the conventional full-duplex wireless-powered communication network (WPCN) in which the AP transmits solely downlink energy signal to all users in the first phase and each user uses its harvested energy to transmit uplink information signal via an additional RF transmitter in the second phase \cite{JuZHangTCOM14}. One typical application example of our considered F-ABCN is described as follows: a WiFi AP simultaneously transmits downlink information via OFDM modulation to its client(s) (e.g., smartphone, laptop) and receives uplink information from multiple domestic IoT devices (e.g., tags, sensors) in smart-home applications. We aim to optimize the throughput performance for a generic F-ABCN in this paper, where its main contributions are summarized as follows:
\begin{itemize}
\item First, to ensure fairness, we formulate a problem to maximize the minimum throughput among all BDs by jointly optimizing the BDs' backscatter time allocation, the BDs' power reflection coefficients, and the FAP's subcarrier power allocation, subject to the LU's throughput requirement and the BDs' harvested-energy constraints, together with other practical constraints. Such a joint optimization problem is practically appealing, since the system performance can benefit from adjusting design parameters in multiple dimensions. However, the formulated problem is non-trivial to solve in general, since the variables are mutually coupled and result in non-convex constraints.
\item Second, for the special case with a single BD, we obtain analytical solutions for the optimal resource allocation, and propose an efficient algorithm for obtaining it based on the Lagrange duality method. The optimal subcarrier power allocation is obtained in semi-closed form that provides useful insights to the optimal design. The convergence and complexity of the algorithm are also analyzed. 
\item Third, for the general case with multiple BDs, we propose an iterative algorithm by leveraging the block coordinated decent (BCD) and successive convex optimization (SCO) techniques. The entire optimization variables are partitioned into three blocks for the BDs' backscatter time allocation, the BDs' power reflection coefficients, and the FAP's subcarrier power allocation, respectively. The three blocks of variables are alternately optimized. However, for the non-convex subcarrier power allocation optimization problem with given backscatter time allocation and power reflection coefficients, we apply the SCO technique to solve it approximately. Also, we show the convergence of the proposed algorithm and analyze its complexity.
\item Fourth, we extend our study by characterizing the throughput region constituting all the Pareto-optimal throughput performance trade-offs among all BDs. Each boundary point of the throughput region is found by solving a sum-throughput maximization problem with a given throughput-profile vector.
\item Last, numerical results show that significant throughput gain is achieved by our proposed joint design, as compared to the benchmark scheme of F-ABCN with equal resource allocation and that of half-duplex AmBC network (H-ABCN) with optimal resource allocation. The BDs-LU throughput trade-off and the BDs' throughput-energy trade-off are revealed as well. Also, the effect of system parameters like the peak power value on the throughput performance is numerically demonstrated.
\end{itemize}

The rest of this paper is organized as follows. Section~\ref{systemmodel} presents the system model for an F-ABCN over ambient OFDM carriers. Section~\ref{formulation} formulates the minimum-throughput maximization problem. Section~\ref{solutionSU} analyzes the joint resource allocation for a single-BD F-ABCN and proposes an optimal algorithm by applying the Lagrange duality method. Section~\ref{solution} proposes an efficient iterative algorithm by applying the BCD and SCO techniques to solve the joint resource allocation problem for a multiple-BD F-ABCN. Section \ref{sec:region} studies the throughput region that characterizes the optimal throughput performance trade-offs among all BDs. Section~\ref{simulation} presents the numerical results to verify the performance of the proposed joint design. Section~\ref{conslusion} concludes this paper.

The main notations in this paper are listed as follows: The lowercase, boldface lowercase, and boldface uppercase letters, e.g., $g$, $\bg$, and $\bG$, denote a scalar, vector, and matrix, respectively. $|g|$ means the operation of taking the absolute value of a scalar $g$. $\bbE[g]$ denotes the statistical expectation of a random variable $g$. $[\bg]^T$ denotes the transpose of a vector $\bg$. The notation $\otimes$ means the convolution operation. $\nabla$ denotes the partial derivative operation. $\calC \calN(0, \sigma^2)$ denotes the circularly symmetric complex Gaussian (CSCG) distribution with zero mean and variance $\sigma^2$. $\calC$ denotes the set of complex numbers. $\calO(\cdot)$ denotes the time complexity order of an algorithm.


\section{System Model}\label{systemmodel}
In this section, we present the system model for an F-ABCN over ambient OFDM carriers. As illustrated in Fig.~\ref{fig:Fig1}, we consider two coexisting communication systems: the legacy communication system which consists of an FAP with two antennas for simultaneous information transmission and reception, respectively, together with its dedicated LU\footnote{We consider the case of a single LU, since the FAP typically transmits to an LU in a short period for practical OFDM systems like WiFi. The analyses and results can be extended to the case of multiple LUs.}, and the AmBC system which consists of the FAP and $M$ ($M \geq 1$) BDs. The FAP transmits OFDM signals to the LU. We are interested in the AmBC system in which each BD transmits its modulated signal back to the FAP over its received ambient OFDM carrier from the FAP. 
Each BD contains a backscatter antenna, a switched load impedance, a micro-controller, an information receiver, an energy harvester, and other modules (e.g., battery, memory, sensing). To transmit information bits, the BD modulates its received ambient OFDM carrier by intentionally switching the load impedance to vary the amplitude and/or phase of its backscattered signal, and the backscattered signal is received and finally decoded by the FAP.

\begin{figure}[t!]
\centering
\includegraphics[width=.99\columnwidth] {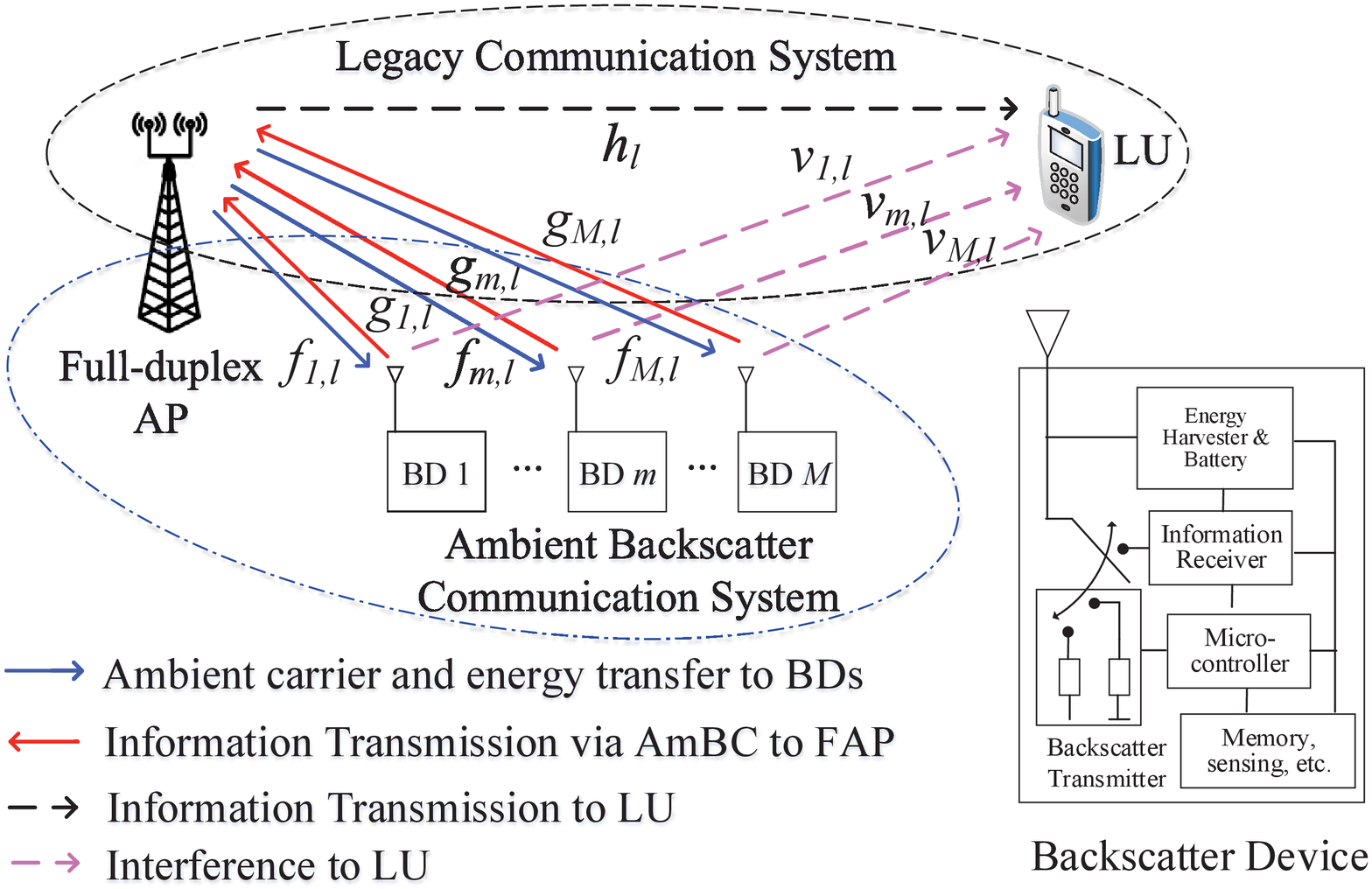}
\caption{System description for an F-ABCN.}
\label{fig:Fig1}
\vspace{-0.2cm}
\end{figure}

The block fading channel model is considered, and the channel block length is assumed to be much longer than the OFDM symbol period. As shown in Fig.~\ref{fig:Fig1}, let $f_{m, l}$ be the $L_{\sf f}$-path forward channel response from the FAP to the $m$-th BD, for $m=1,\ \ldots, \ M$, $g_{m, l}$ be the $L_{\sf g}$-path backward channel response from the $m$-th BD to the FAP, $h_{l}$ be the $L_{\sf h}$-path legacy channel response from the FAP to the LU, and $v_{m, l}$ be the $L_{\sf v}$-path interference channel response from the $m$-th BD to the LU. Let $N (N \geq 1)$ be the number of subcarriers of the transmitted OFDM signals. For the downlink channel from the FAP to the $m$-th BD, we define the frequency response at the $k$-th subcarrier as $F_{m,k} = \sum_{l=0}^{L_{\sf f}-1} f_{m,l}e^{\frac{-j 2\pi k l}{N}}$, for $k = 0,\ldots ,N-1$. Similarly, for the backward channel from the $m$-th BD to the FAP, we define its subcarrier response as $G_{m,k} = \sum_{l=0}^{L_{\sf g}-1} g_{m,l}e^{\frac{-j 2\pi k l}{N}}$; for the interference channel from the $m$-th BD to the LU, we define its subcarrier response as $V_{m,k} = \sum_{l=0}^{L_{\sf v}-1} v_{m,l}e^{\frac{- j 2 \pi k l }{N}}$; and for the legacy channel from the FAP to the LU, we define its subcarrier response as $H_{k} = \sum_{l=0}^{L_{\sf h}-1} h_{l}e^{\frac{- j 2 \pi k l }{N}}$.

We consider a frame-based protocol as shown in Fig.~\ref{fig:Fig2}. The frame duration of $T$ (seconds) is within the channel block length. In each frame consisting of $M$ slots, the FAP simultaneously transmits downlink OFDM signals to the LU, and receives uplink signals backscattered from all BDs in a TDMA manner. The $m$-th slot of time duration $\tau_m T$ (with time proportion $\tau_m$ ($0 \leq \tau_m \leq 1$)) is assigned to the $m$-th BD. Denote the backscatter time allocation vector $\btau=[\tau_1 \ \tau_2 \ \ldots \ \tau_M]^T$. In the $m$-th slot, BD $m$ reflects back a portion of its incident signal for transmitting information to the FAP and harvests energy from the remaining incident signal, and all other BDs only harvest energy from their received OFDM signals.

\begin{figure}[t!]
\vspace{0.1cm}
\centering
\includegraphics[width=.99\columnwidth] {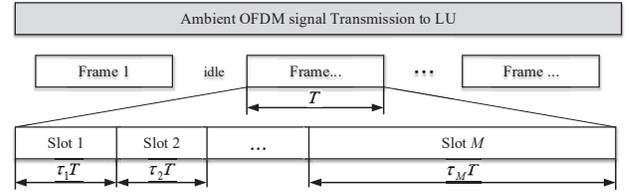}
\caption{Frame-based protocol for an F-ABCN.}
\label{fig:Fig2}
\vspace{-0.4cm}
\end{figure}

Let $S_{m,k}(n) \in \calC$ be the FAP's information symbol at the $k$-th subcarrier, $\forall k$, in the $n$-th OFDM symbol period of the $m$-th slot. After inverse discrete Fourier transform (IDFT) at the FAP, a CP of length $N_{\sf cp}$ is added at the beginning of each OFDM symbol. The transmitted time-domain signal in each OFDM symbol period is
\begin{align}
  s_{m, t} (n) = \frac{1}{N}\sum_{k=0}^{N-1} \sqrt{P_{m,k}} S_{m, k} (n) e^{j 2\pi \frac{kt}{N}},
\end{align}
for the time index $t =0, 1, \ldots, N -1$, where $P_{m,k}$ is the allocated power at the $k$-th subcarrier in the $m$-th slot. Denote the subcarrier power allocation matrix $\bP=[\bp_1 \ \bp_2 \ \ldots \ \bp_M]$, where $\bp_m$ is the subcarrier power allocation vector in the $m$-th slot.

In the $m$-th slot, the incident signal at BD $m$ is $s_{m, t} (n) \otimes f_{m,l}$. From \cite{BoyerSumit14}, due to the impedance discontinuity of the antenna and the load, a proportion $\alpha_m$ ($0 \leq \alpha_m \leq 1$, referred to as the power reflection coefficient) of the incident power is reflected backward, giving rise to the backscattered field, while the remaining $(1-\alpha_m)$ power propagates to the energy-harvesting circuit. For convenience, denote the power reflection coefficient vector $\balpha=[\alpha_1 \ \alpha_2 \ldots \alpha_M]^T$. Let $\eta_m$ ($0 \leq \eta_m \leq 1$), $\forall m$, be the energy-harvesting efficiency constant~\cite{ZhouZhangWirelessCom14,YangBackscatter15,KimDIKimTWC17} of BD $m$. According to the aforementioned energy-harvesting scheme in the proposed protocol and from~\cite{ZhouZhangWirelessCom14}, the total energy harvested by BD $m$ in all slots is thus
\begin{align}
  &E_m (\btau, \alpha_m, \bP)   \label{eq:erergyEm} \\
  & =\eta_m \! \sum_{k=0}^{N-1} \!  |F_{m,k}|^{2} \big[ \tau_{m}  P_{m,k} (1 \!-\! \alpha_m)   \!+\! \sum_{r=1, \ r \neq m}^M  \tau_{r} P_{r, k} \big], \nonumber
\end{align} 
where the first term in the square brackets relates to the harvested energy in the $m$-th slot, and the second term relates to the harvested energy in all other slots.

From the antenna scatterer theorem~\cite{FuschiniFalciaseccaAWP08}, the electronic-magnetic (EM) field backscattered from the $m$-th BD consists of the structural mode (load-independent) component and the antenna mode (load-dependent) component. The former is interpreted as the scattering from the antenna loaded with a reference impedance\footnote{The reference impedance $Z_{\sf ref}$ can be arbitrary, which is typically taken as $0, \infty$, and the antenna impedance $Z_{\sf a}$ for the short-circuit case, the open-circuit case, and the matched circuit case, respectively~\cite{FuschiniFalciaseccaAWP08}.}, which depends on only the antenna's geometry and material. The latter relates to the rest scattering of the antenna, which depends on the specific impedance of the load connected to the antenna. Let $X_m(n) \in \calC$ be the $m$-th BD's information symbol, whose duration is designed to be the same as the OFDM symbol period. We assume that each BD can align the transmission of its own symbol $X_m(n)$ with its received OFDM symbol\footnote{BD can practically estimate the arrival time of OFDM signal by some methods like the scheme that utilizes the repeating structure of CP~\cite{YangLiangZhangPeiTCOM17}.}. The signal backscattered by the $m$-th BD, denoted by $r_{m,t}(n)$, can be written as \cite{BoyerSumit14}
\begin{align}
r_{m,t}(n)  =  s_{m, t} (n) \otimes f_{m,l}(A_{\sf s}-\Gamma_m(n)),
\end{align}
where $A_{\sf s} \in \calC$ is the structural mode component, and the antenna mode component, denoted as $\Gamma_m(n)$, is defined as $\Gamma_m(n) \triangleq -\sqrt{\alpha_m} X_m(n)$ \cite{DarsenaVerdeTCOM17}. Since the structural mode component is fixed for each BD, it can be reconstructed and subtracted from the received signal at the FAP. Hence, for simplicity, we ignore the structural mode component and denote the backscattered signal as $\tilr_{m,t}(n)  = \sqrt{\alpha_m} s_{m, t} (n) \otimes f_{m,l} X_m(n)$ in the sequel.

Since the transmitted downlink signal $s_{m,t}(n)$ is known by the FAP's receiving chain, it can also be reconstructed and subtracted from the received signal. Therefore, the self-interference can be cancelled by using existing digital or analog cancellation techniques \cite{BackFiSigcom15}. For this reason, we assume perfect self-interference cancellation (SIC) at the FAP in this paper. After performing SIC, the received time-domain signal backscattered from the $m$-th BD is given by
\begin{align}
y_{m, t} (n) \!=\! \sqrt{\alpha_m} s_{m, t} (n) \otimes f_{m,l} \otimes g_{m,l} X_m(n) \!+\! w_{m, t} (n), 
\end{align}
where $w_{m, t} (n)$ denotes the additive white Gaussian noise (AWGN) with power $\sigma^2$, i.e., $w_{m,t}(n) \sim \calC \calN(0, \sigma^2)$.

After CP removal and discrete Fourier transform (DFT) at the FAP, the received frequency-domain signal is 
\begin{align}\label{eq:FD_RX-AP}
  &Y_{m, k} (n) = \\
  &\quad \sqrt{P_{m,k}} \sqrt{\alpha_m} F_{m, k} G_{m,k} S_{m, k} (n) X_m (n) + W_{m, k} (n), \nonumber
\end{align}
where the frequency-domain noise $W_{m,k} (n) \sim \calC \calN(0, \sigma^2)$.

The FAP performs maximum-ratio-combining (MRC) to recover the BD symbol $X_m(n)$ as follows,
\begin{align}
  \hatX_m (n) =\frac{1}{N} \sum \limits_{k=0}^{N-1} \frac{Y_{m, k} (n)}{\sqrt{P_{m,k}} \sqrt{\alpha_m} F_{m,k} G_{m,k}S_{m,k}(n)},
\end{align}
and the resulted decoding signal-to-noise-ratio (SNR) is 
\begin{align}
\gamma_m (\alpha_m, \bP) = \frac{\alpha_m}{\sigma ^{2}} \sum \limits_{k=0}^{N-1} |F_{m,k} G_{m,k}|^{2}P_{m,k}.
\end{align}

Hence, the $m$-th BD's throughput\footnote{This paper adopts normalized throughput with unit of bits-per-second-per-Hertz (bps/Hz).} normalized to the frame duration $T$ is 
\begin{align}\label{eq:RateBDm}
  &R_m (\tau_m, \alpha_m, \bp_m)= \nonumber \\
  &\quad \frac{\tau_{m}}{N} \log \left(1+ \frac{\alpha_m}{\sigma^{2}} \sum_{k=0}^{N-1} |F_{m,k} G_{m,k}|^{2}P_{m,k} \right).
\end{align}

Since the backscattered signal is transmitted at the same frequency as the downlink signal in the legacy system, the whole system in Fig.~\ref{fig:Fig1} is indeed a spectrum sharing system~\cite{KangLiangICC17, YCLiangTVT15, HanLiangSRTCOM15, ZhangLiangWCOM16}. The LU receives the superposition of the downlink legacy signal and the backscatter-link signal. Similar to \eqref{eq:FD_RX-AP}, the received frequency-domain signal at the LU can be thus written as follows,
\begin{align}\label{eq:FD_RX-LU}
 & Z_{m, k} (n) = \sqrt{P_{m,k}} H_{k} S_{m, k} (n)  +... \\
 &\sqrt{P_{m,k}} \sqrt{\alpha_m} F_{m, k} V_{m,k} S_{m, k} (n) X_m (n) + \tilW_{m, k} (n), \forall k, m \nonumber
\end{align}
where the frequency-domain noise $\tilW_{m,k} (n) \sim \calC \calN(0, \sigma^2)$.

By treating backscatter-link signal as interference, the total throughput of the LU is given by
\begin{align}\label{eq:RateLU}
  &\tilR(\btau, \balpha, \bP) = \\
  &\frac{1}{N}\sum_{m=1}^M \tau_m \sum_{k=0}^{N-1}  \log \left(1+\frac{|H_{k}|^2 P_{m,k} }{\alpha_m  |F_{m,k}V_{m,k}|^2 P_{m,k} + \sigma^2}\right). \nonumber
\end{align}

\section{Problem Formulation}\label{formulation}


Our objective is to maximize the minimum throughput among all BDs, by jointly optimizing the BDs' backscatter time allocation  (i.e., $\btau$), the BD's power reflection coefficients (i.e., $\balpha$), and the FAP's subcarrier power allocation (i.e., $\bP$). Mathematically, the optimization problem is equivalently formulated as follows,
\begin{subequations}
\label{eq:P1}
\begin{align}
&\underset{Q, \btau, \balpha, \bP}{\max}  \quad Q \\
& \text{s.t.} \quad \frac{\tau_{m}}{N} \!\log\! \!\left(\!1 \!+\ \! \frac{\alpha_m}{\sigma^{2}} \sum_{k=0}^{N-1} |F_{m,k} G_{m,k}|^{2}P_{m,k} \!\right)\!  \!\geq\! Q,  \forall m \label{eq:C1P1}\\
& \!\sum_{m=1}^M\! \frac{\tau_m}{N} \sum_{k=0}^{N-1}  \!\log\! \!\left(\!1 \!+\! \frac{|H_{k}|^2 P_{m,k} }{\alpha_m  |F_{m,k}V_{m,k}|^2 P_{m,k} \!+\! \sigma^2}\!\right)\!  \!\geq\! D \label{eq:C2P1} \\
&\eta_m \!\sum_{k=0}^{N-1}\!  |F_{m,k}|^{2} \!\Big[\! \tau_{m}  P_{m,k} ( 1\!-\! \alpha_m ) \!+\! \!\sum_{r=1, \ r \neq m}^M\!  \tau_r P_{r,k} \!\Big]\! \nonumber \\
& \quad \quad \quad \quad \quad \!\geq\! E_{\min, m}, \quad \forall m \label{eq:C3P1} \\
& \sum_{m=1}^M  \sum_{k=0}^{N-1} \tau_{m}P_{m,k}\leq \barP \label{eq:C4P1}\\
& \sum_{m=1}^M \tau_{m}\leq 1 \label{eq:C5P1} \\
&  \tau_{m}\geq  0, \quad \forall m \label{eq:C7P1}\\
& 0 \leq P_{m,k} \leq  P_{\sf peak}, \quad \forall m, \ k \label{eq:C6P1} \\
&  0 \leq \alpha_m \leq  1, \quad \forall m. \label{eq:C8P1}
\end{align}
\end{subequations}
Note that \eqref{eq:C1P1} is the common-throughput constraint for each BD, \eqref{eq:C2P1} is the LU's requirement of a given minimum throughput $D$; \eqref{eq:C3P1} is each BD's requirement of a given minimum energy $E_{\min, m}$; \eqref{eq:C4P1} is the FAP's maximum (total) transmission-power (i.e., a given value $\barP$) constraint; \eqref{eq:C5P1} is the total backscatter-time constraint, and \eqref{eq:C7P1} is the non-negative constraint for each backscatter time; \eqref{eq:C6P1} is the non-negative and peak-power (i.e., a given value $P_{\sf peak}$) constraint for each subcarrier power; and \eqref{eq:C8P1} is the constraint for each power reflection coefficient.

The above joint optimization problem is practically appealing. On one hand, by properly designing the power reflection coefficients of near BDs, more backscatter time can be allocated to far BDs to further enhance their throughput performance, alleviating the effect of double near-far problem for wireless-powered (backscatter) communication networks \cite{JuZHangTCOM14,KimDIKimTWC17}. On the other hand, by properly allocating subcarrier power at the FAP, better throughput trade-off can be achieved among the BDs and the LU. However, problem \eqref{eq:P1} is challenging to solve, due to the following two reasons. First, the backscatter time allocation variables $\tau_{m}$'s, the power reflection coefficient variables $\alpha_m$'s and the subcarrier power variables $P_{m,k}$'s are all coupled in the constraints \eqref{eq:C1P1}, \eqref{eq:C2P1}, \eqref{eq:C3P1}, and \eqref{eq:C4P1}. Second, the logarithm function in the constraint \eqref{eq:C2P1} is a non-convex function of the subcarrier power variables $P_{m,k}$'s. Therefore, problem \eqref{eq:P1} is non-convex, which is difficult to solve optimally in general.


\section{Joint Resource Allocation in a Single-BD F-ABCN}\label{solutionSU}
To obtain tractable analytical results, in this section, we consider the special case of $M=1$, i.e., a single-BD F-ABCN. For brevity, the subscript $m$ for BD is omitted in the notations, as $m=1$. The transmission power allocation matrix $\bP$, the power reflection coefficient vector $\balpha$ and the backscatter time allocation vector $\btau$ reduce to the vector $\bp=[P_0, P_1, \ \ldots, P_{N-1}]^T$, the scaler $\alpha$ and the constant $\tau=1$ for the BD, respectively. Problem \eqref{eq:P1} is then simplified as follows, 
\begin{subequations}
\label{eq:P2}
\begin{align}
\underset{\alpha, \bp}{\max}  &\quad \frac{1}{N} \!\log\! \!\left(\!1 \!+\ \! \frac{\alpha}{\sigma^{2}} \sum_{k=0}^{N-1} |F_{k} G_{k}|^{2}P_{k} \!\right) \label{eq:ObjP2} \\
\text{s.t.} &\quad  \frac{1}{N} \sum_{k=0}^{N-1}  \!\log\! \!\left(\!1 \!+\! \frac{|H_{k}|^2 P_{k} }{\alpha  |F_{k}V_{k}|^2 P_{k} \!+\! \sigma^2}\!\right)\!  \!\geq\! D \label{eq:C1P2} \\
&\quad \eta \!\sum_{k=0}^{N-1}\!  |F_{k}|^{2} P_{k} ( 1\!-\! \alpha) \geq\! E_{\min} \label{eq:C2P2} \\
& \quad \sum_{k=0}^{N-1} P_{k}\leq \barP \label{eq:C3P2}\\
& \quad 0 \leq P_{k} \leq  P_{\sf peak}, \quad \forall \ k \label{eq:C4P2} \\
&  \quad 0 \leq \alpha \leq  1. \label{eq:C5P2}
\end{align}
\end{subequations}

Since the objective function in \eqref{eq:ObjP2} and the constraint functions in \eqref{eq:C1P2} and \eqref{eq:C2P2} are all monotonically increasing with respect to each individual $P_k$, thus the constraint in \eqref{eq:C3P2} should hold with equality at the optimal power allocation (otherwise, the objective function together with the left-hand-sides (LHSs) of the constraints in \eqref{eq:C1P2} and \eqref{eq:C2P2} can be further increased by increasing some $P_k$'s).

To obtain useful insights, we further assume that the interference from the BD to the LU is negligible, i.e., $\alpha  |F_{k}V_{k}|^2 P_{k} \approx 0$. This assumption is practical, since the interference signal goes through the FAP-to-BD channel fading, the power reflection loss at the BD, and the BD-to-LU channel fading, usually leading to much smaller interference power at the LU compared to the signal directly from the FAP. The general case of non-negligible interference will be studied in Section~\ref{solution}. The optimal $\alpha$ of problem~\eqref{eq:P2} can be obtained by one-dimensional search, and we focus on optimizing the subcarrier power $\bp$ in the rest of this section. Since the logarithm function in \eqref{eq:ObjP2} is monotonically increasing with its argument, problem~\eqref{eq:P2} for given $\alpha$ can be rewritten as 
\begin{subequations}
\label{eq:P2A}
\begin{align}
\underset{\bp}{\max}  \quad & \sum_{k=0}^{N-1} |F_{k} G_{k}|^{2}P_{k} \label{eq:ObjP2A} \\ 
 \text{s.t.} \quad & \frac{1}{N} \sum_{k=0}^{N-1}  \!\log\! \!\left(\!1 \!+\! \frac{|H_{k}|^2 P_{k} }{\sigma^2}\!\right)\!  \!\geq\! D \label{eq:C1P2A} \\
&\eta \!\sum_{k=0}^{N-1}\!  |F_{k}|^{2} P_{k} ( 1\!-\! \alpha) \geq\! E_{\min} \label{eq:C2P2A} \\
& \sum_{k=0}^{N-1} P_{k}= \barP \label{eq:C3P2A}\\
& 0 \leq P_{k} \leq  P_{\sf peak}, \quad \forall \ k. \label{eq:C4P2A}
\end{align}
\end{subequations}
It can be easily checked that problem~\eqref{eq:P2A} is a convex optimization problem with respect to $\bp$, thus can be solved by the Lagrange duality method, as shown as follows.

From \eqref{eq:ObjP2A}, \eqref{eq:C1P2A}, \eqref{eq:C2P2A} and \eqref{eq:C3P2A}, the Lagrangian of problem~\eqref{eq:P2A} is given by
\begin{align}
&\calL(\bp,\lambda,\theta,\mu) = \sum_{k=0}^{N-1} |F_{k} G_{k}|^{2}P_{k} +... \label{eq:Lagrangian}\\
& \qquad  \lambda \left(\frac{1}{N} \sum_{k=0}^{N-1}  \!\log\! \!\left(\!1 \!+\! \frac{|H_{k}|^2 P_{k} }{\sigma^2}\!\right)\!-D\right)+...\nonumber \\
& \theta \left( \eta \! \sum_{k=0}^{N-1}\!  |F_{k}|^{2} P_{k} ( 1\!-\! \alpha) \!-\! E_{\min} \right)- \mu \left( \sum_{k=0}^{N-1} P_{k} \!-\! \barP\right), \nonumber
\end{align}
where $\lambda \geq 0, \ \theta \geq 0$ and $\mu$ denotes the dual variables associated with \eqref{eq:C1P2A}, \eqref{eq:C2P2A}, and \eqref{eq:C3P2A}, respectively. The dual function of problem~\eqref{eq:P2A} is then given by 
\begin{align}\label{eq:dualfunc}
  \calG(\lambda,\theta,\mu) = \underset{0 \leq P_{k} \leq  P_{\sf peak}, \forall k}{\max} \quad \calL(\bp,\lambda,\theta,\mu).
\end{align}
The dual problem of problem~\eqref{eq:P2A} is thus give by $\underset{\lambda \geq 0, \theta \geq 0, \mu}{\min} \quad \calG(\lambda,\theta,\mu)$.

\begin{mythe}\label{the:OptP}
Given $\lambda \geq 0, \theta \geq 0$ and $\mu$, the maximizer of $\calG(\bp, \lambda,\theta,\mu)$ in \eqref{eq:Lagrangian} is given by
\begin{align}\label{eq:OptPk}
  &P_k^{\star} = \min \Big[P_{\sf peak}, \\
  &\quad \left( \frac{\lambda}{N(\mu-|F_k G_k|^2 - \theta \eta |F_k|^2 (1-\alpha))}-\frac{\sigma^2}{|H_k|^2}\right)^{+}\Big], \nonumber
\end{align}
where $(x)^{+}=\max(x, 0)$.
\end{mythe}

\begin{proof}
  Please see Appendix \ref{app:Theorem1}.
\end{proof}

We conclude that $\mu>0$, since Theorem~\ref{the:OptP} implies $P_k^{\star}=0, \forall k$, and the objective value is zero, if $\mu \leq 0$, which is in contradiction with the optimality of $\{P_k^{\star}\}$'s.

From Theorem~\ref{the:OptP}, the optimal solution of problem \eqref{eq:P2A} can be obtained as follows. With $\calG(\lambda,\theta,\mu)$ obtained for each given  pair of $\lambda, \theta$ and $\mu$, the optimal dual variables $\lambda, \theta$ and $\mu$ that minimize $\calG(\lambda,\theta,\mu)$ can then be efficiently obtained by a sub-gradient based algorithm, with the sub-gradient of $\calG(\lambda,\theta,\mu)$ given by 
\begin{subequations}\label{eq:sugradient}
  \begin{align}
\nabla\lambda &=\frac{1}{N} \sum_{k=0}^{N-1}  \!\log\! \!\left(\!1 \!+\! \frac{|H_{k}|^2 P_{k} }{\sigma^2}\!\right)\!-D \\
\nabla \theta &=\sum_{k=0}^{N-1}\!  |F_{k}|^{2} P_{k} ( 1\!-\! \alpha) \!-\! E_{\min} \\
\nabla \mu &= \barP - \sum_{k=0}^{N-1} P_{k}.
\end{align}
\end{subequations}

The overall steps for solving problem \eqref{eq:P2A} are summarized in Algorithm \ref{Algorithm_SU}. Since problem \eqref{eq:P2A} is convex, the proposed Algorithm \ref{Algorithm_SU} is guaranteed to converge \cite{CVXBoyd04}. The computation time of Algorithm \ref{Algorithm_SU} is analyzed as follows. The time complexity of step 3  is $\calO(N)$, and those of step 4 and step 6 are $\calO(1)$. Since only three dual variables, $\lambda, \ \theta, \ \mu$, are updated by the sub-gradient method regardless of the number of BDs, $M$. The time complexity of step 5 is thus $\calO(1)$. As a result, the total time complexity of Algorithm \ref{Algorithm_SU} is $\calO(N)$.

\begin{algorithm}[t!]
\caption{Iterative algorithm for solving problem \eqref{eq:P2A}}\label{Algorithm_SU}
\begin{algorithmic}[1]
\STATE Initialize dual variables $\lambda^{\{0\}}>0, \theta^{\{0\}}>0, \mu^{\{0\}}$, positive step-sizes $\xi_1, \ \xi_2, \ \xi_3$, and small threshold constant $\epsilon=10^{-4}$. Let $i=0$.\\ 
\REPEAT
\STATE Given $\lambda^{\{i\}},\theta^{\{i\}}$ and $\mu^{\{i\}}$, compute $\bp^{\{i\}}$ by using \eqref{eq:OptPk}, and obtain the corresponding dual function value $\calG^{\{i\}}=\calG(\lambda^{\{i\}},\theta^{\{i\}},\mu^{\{i\}})$ as in~\eqref{eq:dualfunc}.
\STATE Compute the sub-gradients $\Delta\lambda^{\{i\}}, \Delta \theta^{\{i\}}$ and $\Delta \mu^{\{i\}}$ given in \eqref{eq:sugradient} by replacing $P_k$ by $P_k^{\{i\}}$.
\STATE Update dual variables
\begin{align}
\lambda^{\{i+1\}}&=\lambda^{\{i\}}+\xi_{1} \Delta\lambda^{\{i\}} \nonumber \\
\theta^{\{i+1\}}&=\theta^{\{i\}}+\xi_{2}\Delta \theta^{\{i\}} \nonumber \\
\mu^{\{i+1\}}&=\mu^{\{i\}}+\xi_{3}\Delta \mu^{\{i\}} \nonumber
\end{align}
\STATE Update iteration index $i=i+1$.
\UNTIL{$  \left( \! \calG^{\{i-1\}} \!-\! \sum \limits_{k=0}^{N-1} |F_{k}|^{2}|G_{k}|^{2}P_k^{\{i-1\}} \! \right) / \calG^{\{i-1\}}  \!<\!  \epsilon$ }
\STATE Obtain the optimal subcarrier power allocation $\bp^{\star}=[P_0^{\{i-1\}}, \ldots, P_{N-1}^{\{i-1\}}]^T$.
\end{algorithmic}
\end{algorithm}

A numerical example is given here to demonstrate the optimal subcarrier power allocation. Fig. \ref{fig:FigSU} depicts the optimal $\bp^{\star}$ that maximizes the BD throughput in a single-BD F-ABCN with $N=16, N_{\sf cp}=8, \barP=1, \eta=0.5, \epsilon=10^{-4}, \sigma^2=-60 \ \text{dBm}, E_{\min}=10  \; \mu\text{J}$, and $D=2  \; \text{bps/Hz}$. We assume independent multi-path Rayleigh fading channels, and the power gains of multiple paths are exponentially distributed. The numbers of channel paths are set as $L_{\sf f}=L_{\sf g}=2$, and $L_{\sf h}=L_{\sf v}=4$. The FAP-to-BD distance, the BD-to-LU distance, and the FAP-to-LU distance are 4, 15, and 15 meters (m), respectively. Other parameters are set as the same as in Section \ref{simulation}. We consider two different peak-power values, $P_{\sf peak}=5 P_{\sf ave}$ and $P_{\sf peak}=10 P_{\sf ave}$ with $P_{\sf ave}=1/N$. For the case of $P_{\sf peak}=5 P_{\sf ave}$, we observe that 98.57\% of the total power is allocated to subcarriers 3 to 6, among which subcarriers 5 and 6 are allocated with peak power of 0.3125, subcarriers 3 and 4 are allocated with power of 0.1751 and 0.1856, respectively, and any other subcarrier's power is negligible. In contrast, for the case of $P_{\sf peak}=10 P_{\sf ave}$, we observe that the power allocation is more concentrated, and 95\% of the total power is allocated to subcarriers 5 and 6. Specifically, only subcarrier 6 is allocated with peak power of 0.6250, and the power at subcarrier 5 is 0.3245, while any other subcarrier's power is much smaller and can be ignored. The allocation criterion can be explained as follows. Under the peak-power constraints, power is allocated with priority to the subcarriers with stronger backscatter-link channel $|F_k G_k|^2$, conditioned on that the LU's throughput constraint and the BD's harvested-energy constraint are satisfied.

\begin{figure} [t]
\vspace{-0.4cm}
	\centering	\includegraphics[width=.99\columnwidth]{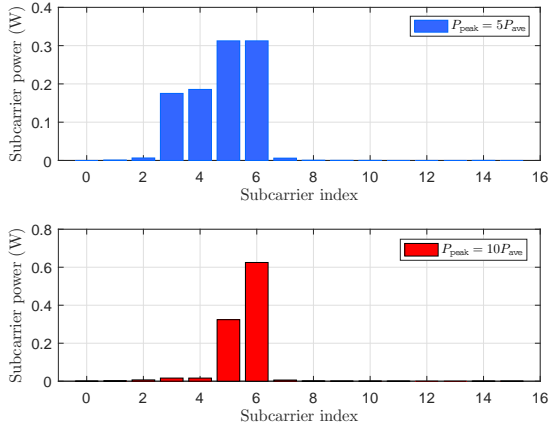}
		\caption{Optimal subcarrier power allocation for different peak-power constraints.} \label{fig:FigSU}
\vspace{-0.6cm}
\end{figure}


\section{Joint Resource Allocation in a Multiple-BD F-ABCN}\label{solution}
In this section, we consider the joint resource allocation in an F-ABCN with multiple BDs. In general, there is no standard method for optimally solving the non-convex optimization problem \eqref{eq:P1} efficiently. Hence, we propose an efficient iterative algorithm to solve it sub-optimally by applying the block coordinate descent (BCD)~\cite{HongLuoBCDSPM17} and successive convex optimization (SCO)~\cite{Beck2010} techniques. In each iteration, we optimize different blocks of variables alteratively. Specifically, for any given power reflection coefficient vector $\balpha$ and subcarrier power allocation matrix $\bP$, we optimize the backscatter time allocation vector $\btau$ by solving a linear programming (LP); for any given backscatter time allocation vector $\btau$ and subcarrier power allocation matrix $\bP$, we optimize the power reflection coefficient vector $\balpha$ by solving a convex problem; and for any given backscatter time allocation vector $\btau$ and power reflection coefficient vector $\balpha$, we optimize the subcarrier power allocation matrix $\bP$ by utilizing the SCO technique and solving an approximated convex problem. After presenting the overall algorithm, we show the convergence of the proposed algorithm and analyze its complexity.
\subsection{Backscatter Time Allocation Optimization}
In iteration $j, j \geq 1$, for given power reflection coefficient vector $\balpha^{\{j\}}$ and subcarrier power allocation matrix $\bP^{\{j\}}$,  the backscatter time allocation vector $\btau$ can be optimized by solving the following problem %
\begin{subequations}
\label{eq:P1BTA}
\begin{align}
&\underset{Q, \btau}{\max}  \quad Q \label{eq:ObjP1BTA} \\
&\text{s.t.}
 \frac{\tau_{m}}{N} \log \left(\!1+ \frac{\alpha_m^{\{j\}}}{\sigma^{2}} \sum_{k=0}^{N-1} |F_{m,k} G_{m,k}|^{2}P_{m,k}^{\{j\}} \!\right)  \! \geq \!Q, \ \forall m \label{eq:C1P1BTA}\\
&\sum_{m=1}^M \frac{ \tau_m }{N} \sum_{k=0}^{N-1}  \log \left( \! 1 \!+ \! \frac{|H_{k}|^2 P_{m,k}^{\{j\}} }{\alpha_m^{\{j\}}  |F_{m,k}V_{m,k}|^2 P_{m,k}^{\{j\}} \!+\! \sigma^2} \! \right)  \!\geq \! D \label{eq:C2P1BTA} \\
&\eta_m\sum_{k=0}^{N-1}  |F_{m,k}|^{2} \Big[ \tau_{m}  P_{m,k}^{\{j\}} (1-\alpha_m^{\{j\}}) +...\nonumber \\
 &\qquad \qquad \quad\sum_{r=1, \ r \neq m}^M  \tau_r P_{r,k}^{\{j\}} \Big] \geq E_{\min, m},\quad \forall m  \label{eq:C3P1BTA}\\
& \sum_{m=1}^M  \sum_{k=0}^{N-1} \tau_{m}P_{m,k}^{\{j\}}\leq \barP \label{eq:C4P1BTA}\\
&\sum_{m=1}^M \tau_{m}\leq 1 \label{eq:C5P1BTA}\\
& \tau_{m}\geq  0, \quad \forall m. \label{eq:C6P1BTA}
\end{align}
\end{subequations}
Since problem~\eqref{eq:P1BTA} is a standard LP, it can be solved efficiently by existing optimization tools such as CVX~\cite{CVXTool2016}. Moreover, it can be verified that either the constraint \eqref{eq:C4P1BTA} or \eqref{eq:C5P1BTA} is met with equality when the optimal $\btau$ is obtained for given $\balpha^{\{j\}}$ and $\bP^{\{j\}}$, since otherwise we can always increase $\tau_m$'s without decreasing the objective value.

\subsection{Reflection Power Allocation Optimization}
For given backscatter time allocation vector $\btau^{\{j\}}$ and subcarrier power allocation matrix $\bP^{\{j\}}$,  the power reflection coefficient vector $\balpha$ can be optimized by solving the following problem
\begin{subequations}
\label{eq:P1RPA}
\begin{align}
&\underset{Q, \balpha}{\max}  \quad Q \\
&\text{s.t.} \;  \frac{\tau_{m}^{\{j\}}}{N} \log \!\left(\!1 \!+\! \frac{\alpha_m}{\sigma^{2}} \sum_{k=0}^{N-1} |F_{m,k}G_{m,k}|^{2}P_{m,k}^{\{j\}} \!\right)  \! \geq \!Q, \forall m \label{eq:C1P1RPA}\\
&\sum_{m=1}^M \frac{ \tau_m^{\{j\}} }{N} \sum_{k=0}^{N-1}  \log \left( \! 1 \!+ \! \frac{|H_{k}|^2 P_{m,k}^{\{j\}} }{\alpha_m  |F_{m,k}V_{m,k}|^2 P_{m,k}^{\{j\}} \!+\! \sigma^2} \! \right)  \!\geq \! D \label{eq:C2P1RPA} \\
&\eta_m\sum_{k=0}^{N-1}  |F_{m,k}|^{2} \Big[ \tau_{m}^{\{j\}}  P_{m,k}^{\{j\}} (1-\alpha_m) +...\nonumber \\
 &\qquad \quad\sum_{r=1, \ r \neq m}^M  \tau_r^{\{j\}} P_{r,k}^{\{j\}} \Big] \geq E_{\min, m},\quad \forall m  \label{eq:C3P1RPA} \\
&\quad  0 \leq \alpha_m \leq  1, \quad \forall m. \label{eq:C4P1RPA}
\end{align}
\end{subequations}
Given $P_{m,k}^{\{j\}}$'s  and $\tau_m^{\{j\}}$'s, \eqref{eq:C1P1RPA} is a convex constraint, while \eqref{eq:C3P1RPA} and \eqref{eq:C4P1RPA} are linear constraints. Moreover, since the LHS of the constraint~\eqref{eq:C2P1RPA} is a decreasing and convex function of $\alpha_m$, this constraint is convex. Hence, problem \eqref{eq:P1RPA} is a convex optimization problem that can also be efficiently solved by CVX~\cite{CVXTool2016}.

\subsection{Subcarrier Power Allocation Optimization}
For given backscatter time allocation vector $\btau^{\{j\}}$ and power reflection coefficient vector $\balpha^{\{j\}}$,  the subcarrier power allocation matrix $\bP$ can be optimized by solving the following problem
\begin{subequations}
\label{eq:P1TPA}
\begin{align}
&\underset{Q, \bP}{\max}\quad Q \\
&\text{s.t.} \; \frac{\tau_{m}^{\{j\}}}{N} \log \!\left(\!1 \!+\! \frac{\alpha_m^{\{j\}}}{\sigma^{2}} \sum_{k=0}^{N-1} |F_{m,k} G_{m,k}|^{2}P_{m,k} \!\right)  \! \geq \!Q, \forall m \label{eq:C1P1TPA} \\
&\sum_{m=1}^M \frac{ \tau_m^{\{j\}} }{N} \sum_{k=0}^{N-1}  \log \left( \! 1 \!+ \! \frac{|H_{k}|^2 P_{m,k}}{\alpha_m^{\{j\}}   |F_{m,k}V_{m,k}|^2 P_{m,k}\!+\! \sigma^2} \! \right)  \!\geq \! D \label{eq:C2P1TPA}\\
&\eta_m\sum_{k=0}^{N-1}  |F_{m,k}|^{2} \Big[ \tau_{m}^{\{j\}}  P_{m,k} (1-\alpha_m^{\{j\}}) +...\nonumber \\
 &\qquad \quad\sum_{r=1, \ r \neq m}^M  \tau_r^{\{j\}} P_{r,k} \Big] \geq E_{\min, m},\quad \forall m  \label{eq:C3P1TPA}\\
&\quad \sum_{m=1}^M  \sum_{k=0}^{N-1} \tau_{m}^{\{j\}} P_{m,k}\leq \barP \label{eq:C4P1TPA}\\
&\quad 0 \leq P_{m,k} \leq  P_{\sf peak}, \quad \forall m, \ k \label{eq:C5P1TPA}
\end{align}
\end{subequations}
Since the constraint function $\tilR(\bP)|_{\btau^{\{j\}}, \balpha^{\{j\}}}$ in \eqref{eq:C2P1TPA} is non-convex with respect to $P_{m, k}$, problem \eqref{eq:P1TPA} is non-convex. Notice that the constraint function $\tilR(\bP)|_{\btau^{\{j\}}, \balpha^{\{j\}}}$ can be rewritten as 
\begin{align}\label{eq:SCA1}
  &\tilR(\bP)|_{\btau^{\{j\}}, \balpha^{\{j\}}} \nonumber \\
   &\!=\! \sum_{m=1}^M \frac{\tau_m^{\{j\}}}{N} \sum_{k=0}^{N-1}  \Big[\!-\!\log \!\left(\! \alpha_m^{\{j\}} |F_{m,k}V_{m,k}|^2 P_{m,k} \!+\! \sigma^2 \!\right)\! +\!...\!  \nonumber \\
  &\quad \log \left( \left(\alpha_m^{\{j\}}  |F_{m,k}V_{m,k}|^2 + |H_{k}|^2 \right) P_{m,k} + \sigma^2 \right) \Big].
\end{align}
To handle the non-convex constraint \eqref{eq:C2P1TPA}, we exploit the SCO technique \cite{Beck2010} to approximate the second logarithm function in \eqref{eq:SCA1}. Recall that any concave function can be globally upper-bounded by its first-order Taylor expansion at any point. Specifically, let $P_{m,k}^{\{j\}}$ denote the subcarrier power allocation matrix in the previous iteration. We have the following concave lower bound at the local point $P_{m,k}^{\{j\}}$
\begin{align}
  &\tilR(\bP)|_{\btau^{\{j\}}, \balpha^{\{j\}}, \bP^{\{j\}}}  \label{eq:SCALB} \!\geq\!\\
  & \sum_{m=1}^M\! \frac{\tau_m^{\{j\}}}{N} \!\sum_{k=0}^{N-1}\!  \Big[ -\! \log \left(\alpha^{\{j\}}  |F_{m,k}V_{m,k}|^2 P_{m,k}^{\{j\}} + \sigma^2\right)\! + \!  ... \nonumber \\
  &\!\log\! \left( \!\left(\!\alpha^{\{j\}}  |F_{m,k}V_{m,k}|^2 \!+\! |H_{k}|^2 \!\right)\! P_{m,k} \!+\! \sigma^2 \right) \! - \nonumber \\
  &\frac{\alpha^{\{j\}}  |F_{m,k}V_{m,k}|^2 (P_{m,k}\!-\! P_{m,k}^{\{j\}})}{\alpha^{\{j\}}  |F_{m,k}V_{m,k}|^2 P_{m,k}^{\{j\}} \!+\! \sigma^2}\Big] \!\triangleq\! \tilR^{\sf lb}(\bP)|_{\btau^{\{j\}}, \balpha^{\{j\}}, \bP^{\{j\}}}.\nonumber
\end{align}

With given local points $\bP^{\{j\}}$ and lower bound $\tilR^{\sf lb}(\bP)|_{\btau^{\{j\}}, \balpha^{\{j\}}, \bP^{\{j\}}}$ in \eqref{eq:SCALB}, by introducing the lower-bound minimum-throughput $Q_{\sf tpa}^{\sf lb}$, problem \eqref{eq:P1TPA} is approximated as the following problem 
\begin{subequations}
\label{eq:P1TPAA}
\begin{align}
&\underset{Q_{\sf tpa}^{\sf lb}, \bP}{\max}  \quad Q_{\sf tpa}^{\sf lb} \\
&\text{s.t.}  \frac{\tau_{m}^{\{j\}}}{N} \log \left(\!1 \!+\!  \frac{\alpha_m^{\{j\}}}{\sigma^{2}} \sum_{k=0}^{N-1} |F_{m,k} G_{m,k}|^{2}P_{m,k} \!\right)  \! \geq \! Q_{\sf tpa}^{\sf lb}, \forall m \label{eq:C1P1TPAA}\\
& \sum_{m=1}^M\! \frac{\tau_m^{\{j\}}}{N} \!\sum_{k=0}^{N-1}\!  \Big[ -\! \log \left(\alpha^{\{j\}}  |F_{m,k}V_{m,k}|^2 P_{m,k}^{\{j\}} + \sigma^2\right)\! + \!  ... \nonumber \\
  &\quad \quad \!\log\! \left( \!\left(\!\alpha^{\{j\}}  |F_{m,k}V_{m,k}|^2 \!+\! |H_{k}|^2 \!\right)\! P_{m,k} \!+\! \sigma^2 \right) \! -... \nonumber \\
  &\quad \quad \frac{\alpha^{\{j\}}  |F_{m,k}V_{m,k}|^2 (P_{m,k}\!-\! P_{m,k}^{\{j\}})}{\alpha^{\{j\}}  |F_{m,k}V_{m,k}|^2 P_{m,k}^{\{j\}} \!+\! \sigma^2}\Big]  \geq D, \label{eq:C2P1TPAA} \\
&\eta_m\sum_{k=0}^{N-1}  |F_{m,k}|^{2} \big[ \tau_{m}^{\{j\}}  P_{m,k} (1-\alpha_m^{\{j\}}) +...\nonumber \\
 &\qquad \quad\sum_{r=1, \ r \neq m}^M  \tau_r^{\{j\}} P_{r,k} \big] \geq E_{\min, m},\quad \forall m  \label{eq:C3P1TPAA} \\
&\sum_{m=1}^M  \sum_{k=0}^{N-1} \tau_{m}^{\{j\}} P_{m,k}\leq \barP \label{eq:C4P1TPAA}\\
& 0 \leq P_{m,k} \leq  P_{\sf peak}, \quad \forall m, \ k. \label{eq:C5P1TPAA}
\end{align}
\end{subequations}
Problem \eqref{eq:P1TPAA} is a convex optimization problem which can also be efficiently solved by CVX~\cite{CVXTool2016}. It is noticed that the lower bound adopted in \eqref{eq:C2P1TPAA} implies that the feasible set of problem \eqref{eq:P1TPAA} is always a subset of that of problem \eqref{eq:P1TPA}. As a result, the optimal objective value obtained from problem \eqref{eq:P1TPAA} is in general a lower bound of that of problem \eqref{eq:P1TPA}.

\subsection{Overall Algorithm}
We propose an overall iterative algorithm for problem \eqref{eq:P1} by applying the BCD technique~\cite{HongLuoBCDSPM17}. Specifically, the entire variables in original problem \eqref{eq:P1} are partitioned into three blocks, i.e., the backscatter time allocation vector $\btau$, power reflection coefficient vector $\balpha$, and subcarrier power allocation matrix $\bP$, which are alternately optimized by solving problem \eqref{eq:P1BTA}, \eqref{eq:P1RPA}, and \eqref{eq:P1TPAA} correspondingly in each iteration, while keeping the other two blocks of variables fixed. Furthermore, the obtained solution in each iteration is used as the input of the next iteration. The details are summarized in Algorithm \ref{AlgorithmP1}.

\begin{algorithm}[t!]
\caption{Block coordinate descent algorithm for solving problem \eqref{eq:P1}}\label{AlgorithmP1}
\begin{algorithmic}[1]
\STATE Initialize $\btau ^{\{0\}}, \ \balpha ^{\{0\}}, \ \bP^{\{0\}}, \ Q^{\{0\}}$ with $\tau ^{\{0\}}_m=\frac{1}{M}, \alpha ^{\{0\}}_m=0.5, P^{\{0\}}_{m, k}=\frac{1}{MN}, \ \forall k, m$, and small threshold constant $\epsilon=10^{-4}$. Let $j=0$. \\
\REPEAT
\STATE Solve problem \eqref{eq:P1BTA} for given $\balpha^{\{j\}}$ and $\bP^{\{j\}}$, and obtain the optimal solution as $\btau^{\{j+1\}}$.
\STATE Solve problem \eqref{eq:P1RPA} for given $\btau^{\{j+1\}}$ and $\bP^{\{j\}}$, and obtain the optimal solution as $\balpha^{\{j+1\}}$.
\STATE Solve problem \eqref{eq:P1TPAA} for given $\btau^{\{j+1\}}$, $\balpha^{\{j+1\}}$, and $\bP^{\{j\}}$, and obtain the optimal solution as $\bP^{\{j+1\}}$.
\STATE  Update iteration index $j=j+1$.
\UNTIL{The increase of the objective value is smaller than $\epsilon$}
\STATE  Return the optimal solution $\btau^{\star}=\btau^{\{j-1\}}$, $\balpha^{\star }=\balpha^{\{j-1\}}$, and $\bP^{\star}= \bP^{\{j-1\}}$.
\end{algorithmic}
\end{algorithm}

\subsection{Convergence and Complexity Analysis}
From~\cite{HongLuoBCDSPM17}, for the classic BCD method, the subproblem for updating each block of variables is required to be solved exactly with optimality in each iteration so as to guarantee its convergence. However, in our proposed Algorithm \ref{AlgorithmP1}, for subcarrier power allocation subproblem \eqref{eq:P1TPA}, we only solve its approximate problem \eqref{eq:P1TPAA} optimally. Thus, the convergence analysis for the classic BCD technique is not applicable to our case, and we prove the convergence of Algorithm \ref{AlgorithmP1} as follows.
\begin{mythe}
Algorithm \ref{AlgorithmP1} is guaranteed to converge.
\end{mythe}

\begin{proof}
First, in step 3 of Algorithm \ref{AlgorithmP1}, since the optimal solution $\btau^{\{j+1\}}$ is obtained for given $\balpha^{\{j\}}$ and $\bP^{\{j\}}$, we have the following inequality on the minimum throughput 
\begin{align}
  Q(\btau^{\{j\}}, \balpha^{\{j\}}, \bP^{\{j\}}) \leq Q(\btau^{\{j+1\}}, \balpha^{\{j\}}, \bP^{\{j\}}). \label{eq:Qinequality1}
\end{align}

Second, in step 4 of Algorithm \ref{AlgorithmP1}, since the optimal solution $\balpha^{\{j+1\}}$ is obtained for given $\btau^{\{j+1\}}$ and $\bP^{\{j\}}$, it holds that
\begin{align}
  Q(\btau^{\{j+1\}}, \balpha^{\{j\}}, \bP^{\{j\}} )\! \leq \! Q(\btau^{\{j+1\}}, \balpha^{\{j+1\}}, \bP^{\{j\}}).\label{eq:Qinequality2}
\end{align}


Third, in step 5 of Algorithm \ref{AlgorithmP1}, it follows that
\begin{align}
  Q(\btau^{\{j+1\}}\!, \! \balpha^{\{j+1\}}\!, \! \bP^{\{j\}})\! &\eqa\! Q^{\sf {lb}, \{j\}}_{\sf tpa} \!(\!\btau^{\{j+1\}}, \! \balpha^{\{j+1\}}, \! \bP^{\{j\}})\! \nonumber \\
  &\!\leb\!  Q^{\sf {lb}, \{j\}}_{\sf tpa} (\btau^{\{j+1\}}, \! \balpha^{\{j+1\}}, \! \bP^{\{j+1\}}) \nonumber \\
  &\!\lec\! Q(\btau^{\{j+1\}}\!, \balpha^{\{j+1\}}\!, \bP^{\{j+1\}}), \label{eq:Qinequality3}
\end{align}
where ($a$) holds since the Taylor expansion in \eqref{eq:SCALB} is tight at given local point, which implies that problem \eqref{eq:P1TPAA} at $\bP^{\{j\}}$ has the same objective function as that of problem \eqref{eq:P1TPA}; ($b$) comes from the fact that $\bP^{\{j+1\}}$ is the optimal solution to problem \eqref{eq:P1TPAA}; and ($c$) holds since the objective value of problem \eqref{eq:P1TPAA} is a lower bound of that of its original problem \eqref{eq:P1TPA}. The inequality in \eqref{eq:Qinequality3} indicates that the objective value is always non-decreasing after each iteration, although an approximated optimization problem \eqref{eq:P1TPAA} is solved to obtain the optimal subcarrier power allocation $\bP$ in each iteration.

From \eqref{eq:Qinequality1}, \eqref{eq:Qinequality2}, and \eqref{eq:Qinequality3}, we further have
\begin{align}
   Q(\!\btau^{\{j\}}, \balpha^{\{j\}}, \bP^{\{j\}})\! \leq Q(\!\btau^{\{j+1\}}, \balpha^{\{j+1\}}, \bP^{\{j+1\}})\!,
\end{align}
which implies that the objective value of problem \eqref{eq:P1} is non-decreasing after each iteration in Algorithm \ref{AlgorithmP1}. It is easy to check that the objective value of problem \eqref{eq:P1} is upper-bounded by some finite positive number. Hence, the proposed Algorithm \ref{AlgorithmP1} is guaranteed to converge. This completes the convergence proof. 
\end{proof}
As will be numerically shown in Section \ref{simulation}, Algorithm \ref{AlgorithmP1} converges typically in a few iterations, which is quite fast for our simulation setup.

Finally, it is noted that the time complexity of Algorithm \ref{AlgorithmP1} is polynomial, since only one LP and two convex optimization problems need to be solved in each iteration. Hence, the proposed Algorithm \ref{AlgorithmP1} can be practically implemented with fast convergence for an F-ABCN with a moderate number of BDs.


\section{Throughput Region Characterization}\label{sec:region}
In this section, we first introduce the throughput region to characterize the optimal throughput performance of all BDs. Then, we formulate an optimization problem to find each boundary point of the throughput region.

The throughput region is defined as follows:
\begin{align}\label{eq:rate_region}
  \calR &\triangleq \underset{\substack{\eqref{eq:C2P1}, \eqref{eq:C3P1}, \eqref{eq:C4P1}, \\ \eqref{eq:C5P1}, \eqref{eq:C7P1}, \eqref{eq:C6P1}, \eqref{eq:C8P1}}}{\bigcup} \; (R_1, \ R_2, \ \ldots, \ R_M).
\end{align}

We apply the technique of throughput-profile vector, which is analogous to the rate-profile vector in~\cite{MohseniCioffJSAC06}, to characterize all the boundary points of the throughput region, where each boundary throughput tuple corresponds to a Pareto-optimal performance trade-off among BDs. Let $R$ denote the sum-throughput achieved by all BDs, i.e., $R=\sum \nolimits_{m=1}^M R_m$. Accordingly, we set  $R_m=\psi_m R$, where the coefficients $\psi_m$'s are subject to $\sum \nolimits_{m=1}^M \psi_m = 1$ and $\psi_m \geq 0, \ \forall m$. With each given throughput-profile vector $\bpsi=[\psi_1 \ \psi_2 \ \ldots \ \psi_M]^T$, the sum-throughput $R$ thus corresponds to a boundary point of the throughput region.

From the definition in \eqref{eq:rate_region}, each boundary point of the throughput region $\calR$ can be obtained by solving the following BD sum-throughput maximization problem with a given throughput-profile vector $\bpsi$.
\begin{subequations}
\label{eq:P3}
\begin{align}
&\underset{R, \btau, \balpha, \bP}{\max}  \quad R \\
& \text{s.t.} \; \frac{\tau_{m}}{N} \!\log\! \!\left(\!1 \!+\ \! \frac{\alpha_m}{\sigma^{2}} \sum_{k=0}^{N-1} |F_{m,k} G_{m,k}|^{2}P_{m,k} \!\right)\!  \!\geq\! \psi_m R ,  \forall m \label{eq:C1P3}\\
&\qquad \eqref{eq:C2P1}, \eqref{eq:C3P1},\eqref{eq:C4P1},\eqref{eq:C5P1},\eqref{eq:C6P1},\eqref{eq:C7P1},\eqref{eq:C8P1}.\label{eq:C2P3}
\end{align}
\end{subequations}

Notice that for given throughput-profile vector $\bm{\psi}$, the only difference between problem \eqref{eq:P3} and problem \eqref{eq:P1} is that each BD's throughput requirement in the constraint \eqref{eq:C1P3} is scaled by a constant $\psi_m$, compared to that in the constraint \eqref{eq:C1P1}. As a result, problem \eqref{eq:P3} can be efficiently solved by an algorithm analogous to Algorithm \ref{AlgorithmP1}, which is omitted herein for brevity.


We give a numerical example to depict the throughput region by setting $M=2, \ N=64$, $D=1$ bps/Hz, $E_{\sf min}=10 \ \mu$J, and $P_{\sf peak}=\frac{20}{MN}$. Set the FAP-to-BD1 distance and FAP-to-BD2 distance as 2.5 m and 4 m, respectively \cite{BackFiSigcom15}. All other parameters are the same as in Section~\ref{simulation}, and the average receive SNR at the FAP is defined as follows
\begin{align}\label{def:SNR}
  \bar{\gamma}=\frac{\barP}{\sigma^2} \sum \limits_{l=0}^{L_{\sf f}-1} \bbE[\left| {{g_{1,l}f_{1, l}}} \right|^2].
\end{align}
Fig.~\ref{fig:FigRegionBD} plots the throughput region of an F-ABCN with two BDs under one channel realization. For the special case of the throughput-profile vector $\bpsi=[0.5 \ 0.5]^T$, the problem \eqref{eq:P3} is equivalent to the max-min problem \eqref{eq:P1}, and the max-min throughputs are plotted as the star-marker points on the corresponding throughput boundary curves. For each SNR, the throughput of BD 1 decreases as that of BD 2 increases, which reveals the throughout trade-off between the two BDs. Specifically, for the given SNR of 20 dB, the maximum throughput of BD 1 is 0.064 bps/Hz (achieved when the throughput of BD 2 is 0 bps/Hz), which is greater than the maximum throughput of BD 2 as 0.054 bps/Hz (achieved when the throughput of BD 1 is 0 bps/Hz). This is because that the FAP-to-BD1 channel is stronger than the FAP-to-BD2 channel, due to shorter distance between the FAP and BD1. Also, we observe that higher throughput is achieved for higher SNR.

\begin{figure} [t!]
	\centering \includegraphics[width=.99\columnwidth]{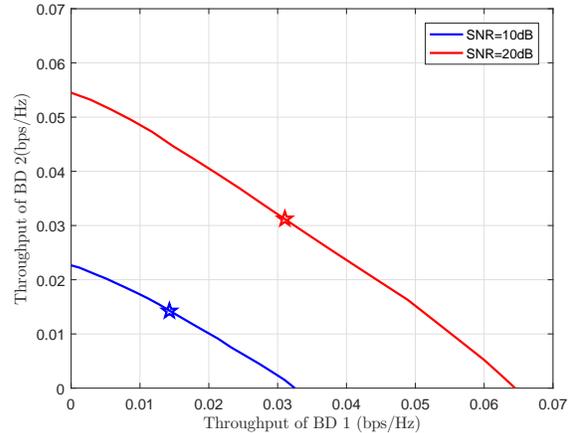}
	\caption{Example of throughput region of an F-ABCN with two BDs.} \label{fig:FigRegionBD}
\vspace{-0.6cm}
\end{figure}


\section{Numerical Results}\label{simulation}
In this section, we provide simulation results to evaluate the performance of the proposed F-ABCN with optimal resource allocation. We consider an F-ABCN with $M=2$ BDs. Suppose that the FAP-to-BD1 distance and FAP-to-BD2 distance are 2.5 m and 4 m, respectively, the FAP (BD1, BD2)-to-LU distances are all 15 m \cite{BackFiSigcom15}. We assume independent Rayleigh fading channels, and the power gains of multiple paths are exponentially distributed. For each channel link, its first-path channel power gain is assumed to be $10^{-3} d^{-2}$, where $d$ denotes the transmitter-to-receiver distance in m. Let the number of paths $L_{\sf f}=L_{\sf g}=4$, $L_{\sf h}=8$, and $L_{\sf v}=6$ \cite{YangLiangZhangPeiTCOM17}. Other parameters are set as $N=64, N_{\sf cp}=16, \barP=1, \epsilon=10^{-4}$, and $\eta_m=0.5, \forall m$. The average receive SNR at the FAP is defined in \eqref{def:SNR}. Let $E_{\min, 1}=E_{\min, 2}=E_{\min}$.  The FAP symbols $S_{m,k}(n)$'s and the BD symbols $X_m(n)$'s are all independently random variables and follow the capacity-achieving distribution, i.e., $S_{m,k}(n) \in \calC \calN(0,1), \ X_m(n) \in \calC \calN(0,1), \ \forall m, k$. The convex subproblems \eqref{eq:P1BTA}, \eqref{eq:P1RPA} and \eqref{eq:P1TPAA} are efficiently solved by the the CVX tool~\cite{CVXTool2016}. All results are obtained based on 100 random channel realizations.

\vspace{-0.2cm}
\subsection{Benchmark Schemes}
For performance comparison, we consider two benchmark schemes. The first one is the case of an F-ABCN with equal resource allocation, in which the backscatter time and subcarrier power are equally allocated, i.e., $\tau_m=\frac{1}{M}, P_{m,k}=P_{\sf ave}=\frac{1}{MN}$, and all BDs adopt a common power reflection coefficient optimized via CVX.

The second benchmark is the case of a half-duplex AmBC network (H-ABCN), in which a half-duplex access point (HAP) first transmits dedicated OFDM signal to the LU to satisfy its throughput constraint in the first phase of time proportion $\tau_0 \ (0 \leq \tau_0 \leq 1)$, then sends dedicated OFDM signal to receive backscattered information from $M$ BDs in a TDMA manner in the second phase (i.e., slot $1, \ \ldots, \ M$). In the first phase, all BDs harvest energy from their received signals. In the $m$-th slot with time proportion $\tau_m, \ 0 \leq \tau_m \leq 1$, for $m=1, \ \ldots, \ M$, of the second phase, BD $m$ reflects a portion (with power proportion $\alpha_m$) of its incident signal for information transmission and harvests energy from the remaining incident signal, and all other BDs only harvest energy. For convenience, we define the augmented backscatter time allocation vector ${\bm{\widetilde{\tau}}}=[\tau_0 \ \tau_1 \ \tau_2 \ \ldots \ \tau_M]^T$, and the augmented subcarrier power allocation matrix ${\bm{\widetilde{P}}}=[\bp_0 \ \bp_1 \ \bp_2 \ \ldots \ \bp_M]$, where $\bp_0$ and $\bp_m$ are the subcarrier power allocation vectors in the first phase and the $m$-th slot of the second phase, respectively. The corresponding minimum throughput among all BDs, denoted as $\tilQ$, can be maximized by solving the following problem
\begin{subequations}
\label{eq:PB2}
\begin{align}
&\underset{\tilQ, {\bm{\widetilde{\tau}}}, \balpha, {\bm{\widetilde{P}}}}{\max}  \quad \tilQ \\
& \text{s.t.} \quad \frac{\tau_{m}}{N} \!\log\! \!\left(\!1 \!+\ \! \frac{\alpha_m}{\sigma^{2}} \sum_{k=0}^{N-1} |F_{m,k} G_{m,k}|^{2}P_{m,k} \!\right)\!  \!\geq\! \tilQ,  \forall m \label{eq:C1PB2}\\
& \frac{\tau_0}{N} \sum_{k=0}^{N-1}  \!\log\! \!\left(\!1 \!+\! \frac{|H_{k}|^2 P_{0, k} }{\sigma^2}\!\right)\!  \!\geq\! D \label{eq:C2PB2} \\
&\eta_m  \sum_{k=0}^{N-1}  |F_{m,k}|^{2} \Big[ \tau_0 P_{0,k}  + \tau_{m}  P_{m,k} ( 1- \alpha_m ) +... \nonumber \\
& \quad \quad \quad \quad \quad \!\sum_{r=1, \ r \neq m}^M  \tau_r P_{r,k} \Big] \geq E_{\min, m}, \quad \forall m \label{eq:C3PB2} \\
& \tau_0 \sum_{k=0}^{N-1}  P_{0,k}+ \sum_{m=1}^M  \sum_{k=0}^{N-1} \tau_{m}P_{m,k}\leq \barP \label{eq:C4PB2}\\
& \tau_0 + \sum_{m=1}^M \tau_{m}\leq 1 \label{eq:C5PB2} \\
&  \tau_{0}\geq  0, \ \tau_{m}\geq  0, \quad \forall m \label{eq:C6PB2}\\
& 0 \leq P_{0,k} \leq  P_{\sf peak}, \quad 0 \leq P_{m,k} \leq  P_{\sf peak}, \quad \forall m, \ k \label{eq:C7BPB2} \\
&  0 \leq \alpha_m \leq  1, \quad \forall m. \label{eq:C8PB2}
\end{align}
\end{subequations}
Notice that the above problem \eqref{eq:PB2} has the same structure as problem \eqref{eq:P1}, thus it can be efficiently solved by an algorithm similar to Algorithm \ref{AlgorithmP1}.


\vspace{-0.2cm}
\subsection{Simulation Results}

Fig.~\ref{fig:FigSim1} plots the max-min throughput of all BDs versus the LU's throughput requirement $D$ under different SNRs $\bar{\gamma}$'s, for the proposed F-ABCN and both benchmark schemes. As in \cite{BoyerSumit14,ZhouZhangWirelessCom14}, we fix $E_{\min}=10 \  \mu$J and $P_{\sf peak} = 20 P_{\sf ave}$. As expected, the max-min throughput decreases as $D$ increases, which reveals the throughput trade-offs between the BDs and the LU. Moreover, compared to both benchmark schemes, we observe that the max-min throughput performance is significantly enhanced by using the proposed joint design. For the case of $D \leq 2$ bps/Hz and 20 dB SNR, the max-min throughput for the proposed F-ABCN with optimal resource allocation is 0.0255 bps/Hz, which is increased by 100.8\% compared to the benchmark of F-ABCN with equal resource allocation, and by 116.1\% compared to the benchmark of H-ABCN with optimal resource allocation. This significant performance gain justifies the advantages of the proposed F-ABCN over the H-ABCN benchmark, although the FAP in an F-ABCN requires higher processing complexity due to the SIC operation. Also, higher max-min throughput is achieved when the receive SNR at the FAP is higher. For an F-ABCN with 20 MHz bandwidth and 20 dB SNR, the achieved max-min BD throughput is around 502 Kbps, when the required LU throughput is no higher than 30 Mbps.

\begin{figure} [!t]
	\centering	\includegraphics[width=.99\columnwidth]{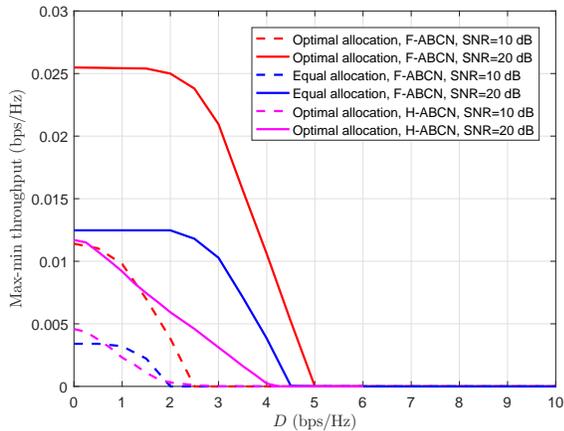}
	\caption{Max-min throughput versus LU's throughput requirement at different SNRs.} \label{fig:FigSim1}
\vspace{-0.4cm}
\end{figure}

\begin{figure} [t!]
	\centering \includegraphics[width=.99\columnwidth]{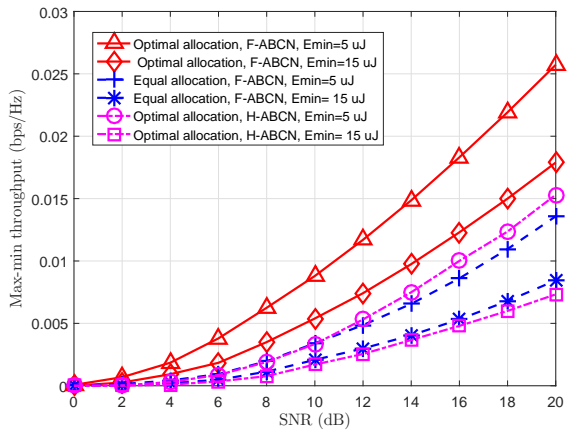}
\caption{Max-min throughput versus SNR with different harvested-energy constraints.} \label{fig:FigSim2aa}
\vspace{-0.6cm}
\end{figure}


Fig.~\ref{fig:FigSim2aa} plots the max-min throughput versus the SNR under different BDs' energy requirements $E_{\min}$'s, for the proposed F-ABCN and both benchmark schemes. We fix the LU's throughput requirement $D=1$ bps/Hz and the subcarrier peak power $P_{\sf peak} = 20 P_{\sf ave}$. First, we observe that the proposed joint design achieves significant throughput gain as compared to the benchmark schemes. For the case of $\bar{\gamma}=20$ dB and $E_{\min}=5  \ \mu$J, the proposed F-ABCN achieves 70\% throughput improvement, compared to the benchmark of H-ABCN. Second, higher throughput is achieved for lower harvested-energy requirement $E_{\min}$ with given $P_{\sf peak}$, which reveals the BDs' throughput-energy trade-off.


Fig.~\ref{fig:FigSim3aa} plots the max-min throughput versus the receive SNR for different subcarrier peak-power values $P_{\sf peak}$'s, for the proposed F-ABCN and both benchmarks. We fix $E_{\min}=10 \  \mu$J and $D=1$ bps/Hz. As compared to the benchmark schemes, the max-min throughput is significantly enhanced by the proposed F-ABCN. For the case of $\bar{\gamma}=20$ dB and $P_{\sf peak}=5 P_{\sf ave}$, the max-min throughput for the proposed F-ABCN is 0.018 bps/Hz, which is increased by 55.2\% compared to the benchmark of F-ABCN with equal resource allocation, and by 73.1\% compared to the H-ABCN benchmark. Also, higher max-min throughput is obtained for higher peak-power value $P_{\sf peak}$.

\begin{figure} [t]
	\centering	\includegraphics[width=.99\columnwidth]{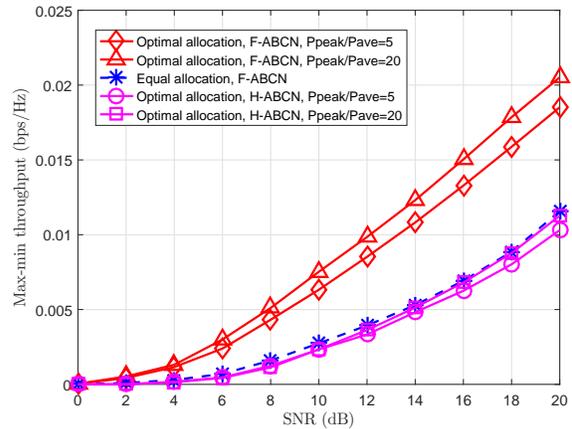}
		\caption{Max-min throughput versus SNR with different peak-power constraints.} \label{fig:FigSim3aa}
\vspace{-0.6cm}
\end{figure}

Finally, we study the convergence performance of the proposed Algorithm \ref{AlgorithmP1} that solves the general optimization problem \eqref{eq:P1} for an F-ABCN with multiple BDs. Fig.~\ref{fig:FigSim4Convergence} depicts the average convergence behavior of Algorithm \ref{AlgorithmP1}. It is observed that this algorithm takes about 5 iterations to converge. The converged average max-min throughput is 0.02028 bps/Hz. To verify that the global max-min throughput is achieved, we compare the obtained max-min throughput with that by an exhaustive search, which is equal to 0.0202 bps/Hz. Thus, Algorithm \ref{AlgorithmP1} does achieve the global optimality of max-min throughput within a guaranteed error of $8 \times 10^{-5}$, which is smaller than the threshold $\epsilon=10^{-4}$ set in the simulation.
\begin{figure} [t]
	\centering \includegraphics[width=.99\columnwidth]{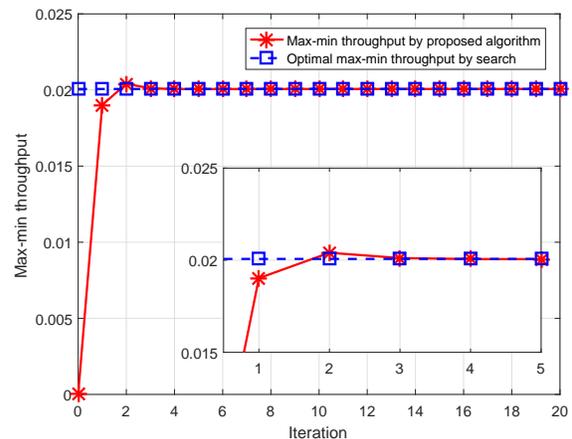}
		\caption{Convergence behavior of Algorithm \ref{AlgorithmP1}.} \label{fig:FigSim4Convergence}
\vspace{-0.6cm}
\end{figure}

\vspace{-0.1cm}
\section{Conclusion}\label{conslusion}
This paper has investigated a full-duplex AmBC network (F-ABCN) over ambient OFDM carriers. The minimum throughput among all BDs is maximized by jointly optimizing the BDs' backscatter time allocation, the BDs' power reflection coefficients, and the FAP's subcarrier power allocation. Analytical solutions are first obtained for the optimal resource allocation in a single-BD F-ABCN. Then, for a multiple-BD F-ABCN, by utilizing the block coordinated decent and successive convex optimization techniques, an efficient iterative algorithm is proposed for solving the non-convex joint optimization problem, which is guaranteed to converge to at least a locally optimal solution. The throughput region is introduced to characterize all the Pareto-optimal throughput performance trade-offs among all BDs. Numerical results show that significant throughput gains are achieved as compared to the benchmark scheme of the F-ABCN with equal resource allocation and that of the half-duplex AmBC network (H-ABCN) with optimal resource allocation, due to the proposed multi-dimensional resource allocation joint optimization and the efficient full-duplex operation at the FAP. The BDs' throughput-energy trade-off and the throughput trade-off between the BDs and the LU are also revealed. This work can be further extended to the cases of multiple LUs, imperfect self-interference cancellation at the FAP, and/or imperfect channel state information, etc.


\appendices
\vspace{-0.15cm}
\section{Proofs of Theorem~\ref{the:OptP}}\label{app:Theorem1}
\begin{proof}
  Given $\alpha$, the Lagrangian in \eqref{eq:Lagrangian} is
  \begin{align}
&\calL(\bp,\lambda,\theta,\mu) = \sum_{k=0}^{N-1} \calL_k (P_k,\lambda,\theta,\mu) \!-\! \lambda D \!-\! \theta E_{\min} \!+\! \mu \barP, \nonumber 
\end{align}
where $ \calL_k (P_k,\lambda,\theta,\mu), \ k=0, 1, \ldots, N-1$ is given by
\begin{align}
&\calL_k(P_k,\lambda,\theta,\mu) = |F_{k} G_{k}|^{2}P_{k} +... \nonumber \\ 
& \qquad  \frac{\lambda}{N} \!\log\! \!\left(\!1 \!+\! \frac{|H_{k}|^2 P_{k} }{\sigma^2}\!\right)+\theta \eta   |F_{k}|^{2} P_{k} ( 1\!-\! \alpha) - \mu P_{k}. \nonumber
\end{align}
Given $\lambda \geq 0, \theta \geq 0$ and $\mu$, the dual function $\calG(\lambda,\theta,\mu)$ in \eqref{eq:dualfunc} can be obtained by maximizing individual $ \calL_k (P_k,\lambda,\theta,\mu), \ k=0, 1, \ldots, N-1$, subject to \eqref{eq:C4P2A}, as $\calL_k(P_k,\lambda,\theta,\mu)$ is only determined by $P_k$, by solving the following problem, for $k=0, \, \ldots,N-1$, 
\begin{subequations}
\label{eq:LagragianOptk}
\begin{align}
&\underset{P_k}{\max}  \quad \calL_k (P_k,\lambda,\theta,\mu) \label{eq:ObjLOptk} \\ 
& \text{s.t.} \quad 0 \leq P_{k} \leq  P_{\sf peak}. \label{eq:C1ObjLOptk}
\end{align}
\end{subequations}

The maximizer of the dual function $\calG(\lambda,\theta,\mu)$ with given $\lambda, \theta$ and $\mu$ can be obtained by setting $\frac{\partial \calL_k}{\partial P_{k}}=0$, from which we have the following equality
\begin{align}\label{eq:KKTDeriv}
  |F_{k} G_{k}|^{2} \!+\! \frac{\lambda}{N(P_k  \!+\! \sigma^2/|H_k|^2)} \!+\! \theta  \eta  |F_{k}|^{2} - \mu \!=\!0.
\end{align}
Since $0 \leq P_k \leq P_{\sf peak}$, the optimal subcarrier power is given from \eqref{eq:KKTDeriv} as in \eqref{eq:OptPk}. This completes the proof.
\end{proof}

\vspace{-0.15cm}
 \renewcommand{\baselinestretch}{0.92}
\bibliography{IEEEabrv,reference1804}
\bibliographystyle{IEEEtran}

\end{document}